\documentclass[11pt]{article}

\usepackage{amsmath,amsfonts,amssymb,amsthm}
\usepackage{amsthm}
\usepackage{graphicx,color}
\usepackage{framed}
\usepackage{enumitem}
\usepackage{thmtools}
\usepackage{thm-restate}
\usepackage{xspace}
\usepackage{bm}
\usepackage{todonotes}
\usepackage{footnote}
\usepackage{mathtools}
\usepackage{mathrsfs}
\usepackage{todonotes}
\usepackage[most]{tcolorbox}
\usepackage{footnote}
\usepackage{rotating}
\usepackage[pdftex, plainpages = false, pdfpagelabels, 
bookmarks=true,
bookmarksopen = true,
bookmarksnumbered = true,
breaklinks = true,
linktocpage,
pagebackref,
colorlinks = true,  
linkcolor = blue,
urlcolor  = blue,
citecolor = red,
anchorcolor = green,
hyperindex = true,
hyperfigures
]{hyperref} 
\usepackage{tabularx}
\usepackage{tikz} 
\usetikzlibrary{calc}
\usepackage{thm-autoref}
\usepackage[nameinlink]{cleveref}
\usepackage[left=2cm, right=2cm, top=2cm]{geometry}

\newtheorem{theorem}{Theorem}
\newtheorem{lemma}{Lemma}

\newtheorem{observation}{Observation}

\newcommand{\lr}[1]{\left( #1\right)}

\newcommand{\LR}[1]{\left\{ #1\right\}}

\newcommand{\Oh}{\mathcal{O}}

\newcommand{\cH}{\mathcal{H}}

\newcommand{\cov}{\mathsf{cov}}

\newcommand{\ex}{\mathbb{E}}

\newcommand{\vs}{V_{\sigma}}

\newcommand{\cova}{\cov_{\alpha}}

\newcommand{\OPT}{\textup{\textsf{OPT}}\xspace}

\newcommand{\maxprob}{\textup{\textsc{Max} $\alpha$-\textsc{FCGP}}\xspace}

\newcommand{\minprob}{\textup{\textsc{Min} $\alpha$-\textsc{FCGP}}\xspace}

\title{FPT Approximation and Subexponential Algorithms for Covering Few or Many Edges \footnote{F. V. Fomin and P. A. Golovach are supported by the Research Council of Norway via the project BWCA (grant no. 314528). T. Inamdar is supported by European Research Council (ERC) under the European Union’s Horizon 2020 research and innovation programme (LOPRE grant no. 819416). T. Koana is supported by the DFG project DiPa (NI 369/21).}}

\author{
	Fedor V. Fomin\thanks{
		Department of Informatics, University of Bergen, Norway.}
	\and
	Petr A. Golovach \addtocounter{footnote}{-1}\footnotemark{}
	\and
	Tanmay Inamdar\addtocounter{footnote}{-1}\footnotemark{}
	\and
	Tomohiro Koana\thanks{Algorithmics and Computational Complexity, Technische Universit\"at Berlin, Germany}
	}
	\date{}

\begin{document}

\maketitle

\begin{abstract}
	We study the \textsc{$\alpha$-Fixed Cardinality Graph Partitioning ($\alpha$-FCGP)} problem, the generic local graph partitioning problem introduced by Bonnet et al. [Algorithmica 2015]. 
	In this problem, we 
are given a graph $G$, two numbers $k,p$ and $0\leq\alpha\leq 1$, the question is whether there is a set $S\subseteq V$ of size $k$ with a specified coverage function $\cova(S)$ at least $p$ (or at most $p$ for the minimization version).
	The coverage function $\cova(\cdot)$ counts edges with exactly one endpoint in $S$ with weight $\alpha$ and edges with both endpoints in $S$ with weight $1 - \alpha$. 	$\alpha$-FCGP generalizes a number of fundamental graph problems such as \textsc{Densest $k$-Subgraph}, \textsc{Max $k$-Vertex Cover}, and \textsc{Max $(k,n-k)$-Cut}.
	
A natural question in the study of  $\alpha$-FCGP is whether the algorithmic results known for its special cases, like \textsc{Max $k$-Vertex Cover}, could be extended to more general settings. One of the simple but powerful methods for obtaining parameterized approximation [Manurangsi, SOSA 2019] and subexponential algorithms [Fomin et al. IPL 2011] for  \textsc{Max $k$-Vertex Cover} is based on the greedy vertex degree orderings.  The main insight of our work is that the idea of greed vertex degree ordering could be used to design fixed-parameter  approximation schemes (FPT-AS) for $\alpha > 0$ and the subexponential-time algorithms for the problem on apex-minor free graphs for maximization with $\alpha > 1/3$ and minimization with $\alpha < 1/3$.
\end{abstract}


\section{Introduction}

In this work, we study a broad class of problems called \textsc{$\alpha$-Fixed Cardinality Graph Partitioning ($\alpha$-FCGP)}, originally introduced by Bonnet et al.~\cite{BEPT15} \footnote{Bonnet et al.~\cite{BEPT15} called the problem `local graph partitioning problem', however we adopt the nomenclature from Koana et al.~\cite{KKNS22}.}.
The input is a graph $G = (V, E)$, two non-negative integers $k, p$, and a real number $0 \le \alpha \le 1$.
The question is whether there is a set $S \subseteq V$ of size exactly $k$ with $\cova(S) \ge p$ ($\cova(S) \le p$ for the minimization variant), where
\begin{align*}
	\cova(S) \coloneqq (1-\alpha) \cdot m(S) + \alpha \cdot m(S, V\setminus S).
\end{align*}
Here, $m(S)$ is the number of edges with both endpoints in $S$, and $m(S, V \setminus S)$ is the number of edges with one endpoint in $S$ and other in $V \setminus S$.
We will call the maximization and minimization problems \maxprob{} and \minprob{}, respectively.
This problem generalizes many problems, namely, \textsc{Densest $k$-Subgraph} (for $\alpha = 0$), \textsc{Max $k$-Vertex Cover}\footnote{This is problem is also referred to as \textsc{Partial Vertex Cover}.} (for $\alpha = 1/2$), \textsc{Max $(k, n - k)$-Cut} (for $\alpha = 1$), and their minimization counterparts.

Although there are plethora of publications that study these special cases, the general \textsc{$\alpha$-FCGP} has not received much attention, except for the work of Bonnet et al.~\cite{BEPT15}, Koana et al.~\cite{KKNS22}, and Schachnai and Zehavi \cite{SZ17}.
In this paper, we aim to demonstrate the wider potential of the existing algorithms designed for specific cases, such as \textsc{Max $k$-Vertex Cover}, by presenting an algorithm that can handle the more general problem of \textsc{$\alpha$-FCGP}.
Algorithms for these specific cases often rely on greedy vertex degree orderings.
For instance, Manurangsi~\cite{Manurangsi19}, showing that a $(1 - \varepsilon)$-approximate solution can be found in the set of $\Oh(k / \varepsilon)$ vertices with the largest degrees, gave a $(1 - \varepsilon)$-approximation algorithm for \textsc{Max $k$-Vertex Cover} that runs in time $(1/\varepsilon)^{\Oh(k)} \cdot n^{\Oh(1)}$.
Fomin et al.~\cite{FominLRS11} gave a $2^{\Oh(\sqrt{k})} \cdot n^{\Oh(1)}$-time algorithm for \textsc{Max $k$-Vertex Cover} on apex-minor graphs via bidimensionality arguments, by showing that an optimal solution $S$ is adjacent to every vertex of degree at least $d + 1$, where $d$ is the minimum degree over vertices in $S$.
In this work, we will give approximation algorithms as well as subexponential-time algorithms for apex-minor free graphs exploiting the greedy vertex ordering.

For approximation algorithms, we will show that both \maxprob{} and \minprob{} admit \emph{FPT Approximation Schemes} (FPT-AS) for $\alpha > 0$, i.e., there is an algorithm running in time $(\frac{k}{\varepsilon \alpha})^{\Oh(k)} \cdot n^{\Oh(1)}$ that finds a set $S$ of size $k$ with $\cova(S) \ge (1 - \varepsilon) \cdot \OPT$ (or $\cova(S) \le (1 + \varepsilon) \cdot \OPT$ for the minimization variant), where $\OPT$ denotes the optimal value of $p$.
Previously, the special cases were known to admit FPT approximation schemes; see \cite{Marx08,GuptaLL18focs,GuptaLL18soda,Manurangsi19} for $\alpha = 1/2$ and \cite{BEPT15} for $\alpha = 1$. 
In particular, the state-of-the-art running time for \maxprob{} with $\alpha = 1/2$ is the aforementioned algorithm of Manurangsi that runs in time $(1/\varepsilon)^{\Oh(k)} \cdot n^{\Oh(1)}$ for maximization (also for the minimization variant). We generalize this argument for $\alpha \ge 1/3$, leading to a faster FPT-AS for \maxprob{} in this range. For $\alpha = 0$, the situation is more negative; \maxprob{} (namely, \textsc{Densest $k$-Subgraph}) does not admit any $o(k)$-approximation algorithm with running time $f(k) \cdot n^{\Oh(1)}$ under the Strongish Planted Clique Hypothesis \cite{ManurangsiRS21}.
\minprob{} is also hard to approximate when $\alpha = 0$ since it encompasses \textsc{Independent Set} as a special case for $p = 0$.

Next, we discuss the regime of subexponential-time algorithms.
Amini et al.~\cite{AFS11} showed that \textsc{Max $k$-Vertex Cover} is FPT on graphs of bounded degeneracy, including planar graphs, giving a $k^{\Oh(k)} \cdot n^{\Oh(1)}$-time algorithm.
They left it open whether it can be solved in time $2^{o(k)} \cdot n^{O(1)}$.
This was answered in the affirmative by Fomin et al.~\cite{FominLRS11}, who showed that \textsc{Max $k$-Vertex Cover} on apex-minor free graphs can be solved in time $2^{\Oh(\sqrt{k})} \cdot n^{\Oh(1)}$ time.
Generalizing this result, we give a $2^{\Oh(\sqrt{k})} \cdot n^{\Oh(1)}$-time algorithm for \maxprob{} with $\alpha > 1/3$ and \minprob{} with $\alpha < 1/3$.
The complexity landscape of \maxprob{} with $\alpha < 1/3$ (and \minprob{} with $\alpha > 1/3$) is not well understood.
It is a long-standing open question whether \textsc{Densest $k$-Subgraph} on planar graphs is NP-hard~\cite{CorneilP84}.
Note that the special case \textsc{Clique} is trivially polynyouomial-time solvable on planar graphs because a clique on 5 vertices does not admit a planar embedding.

\paragraph*{Further related work.}

As mentioned, special cases of \textsc{$\alpha$-FCGP} when $\alpha \in \{ 0, 1/2, 1 \}$ have been extensively studied.
For instance, the W[1]-hardness for the parameter $k$ has been long known for these special cases \cite{Cai08,DowneyEFPR03,GNW07}.
Both \maxprob{} and \minprob{} are actually W[1]-hard for every $\alpha \in [0, 1]$ with the exception $\alpha \ne 1/3$, as can be seen from a parameterized reduction from \textsc{Clique} and \textsc{Independent Set} on regular graphs.
Note that \textsc{$\alpha$-Fixed Cardinality Graph Partitioning} becomes trivial when $\alpha = 1/3$ because $\cova(S) = \frac{1}{3} \cdot \sum_{v \in S} d(v)$ for any $S \subseteq V$ where $d(v)$ is the degree of $v$.

Bonnet et al.~\cite{BEPT15} gave a $(\Delta k)^{2k} \cdot n^{\Oh(1)}$-time algorithm for \textsc{$\alpha$-FCGP} where $\Delta$ is the maximum degree.
They also gave an algorithm with running time $\Delta^k \cdot n^{\Oh(1)}$ for  \maxprob{} with $\alpha > 1/3$ and \minprob{} with $\alpha < 1/3$.
This result was strengthened by Schachnai and Zehavi \cite{SZ17}; they gave a $4^{k + o(k)} \Delta^k \cdot n^{\Oh(1)}$-time algorithm for any value of~$\alpha$.
Koana et al.~\cite{KKNS22} showed that \maxprob{} admits polynomial kernels on sparse families of graphs when $\alpha > 1/3$.
For instance, \maxprob{} admits a $k^{\Oh(d)}$-sized kernel where $d$ is the degeneracy of the input graph.
They also showed analogous results for \minprob{} with $\alpha < 1/3$.

\paragraph*{Preliminaries.}

For an integer $n$, let $[n]$ denote the set $\{ 1, \cdots, n \}$. 

We use the standard graph-theoretic notation and refer to the textbook of Diestel~\cite{Diestel12} for undefined notions.  
In this work, we assume that all graphs are simple and undirected. 
For a graph $G$ and a vertex set $S$, let $G[S]$ be the subgraph of $G$ induced by $X$. 
For a vertex $v$ in $G$, let $d(v)$ be its \emph{degree}, i.e., the number of its neighbors.
For vertex sets $X, Y$, let $m(X) \coloneq |\{ uv \in E \mid u, v \in X \}|$ and $m(X, Y) \coloneq |\{ uv \in E \mid u \in X, v \in Y \}|$. 
In this work, an optimal solution for \maxprob (and \minprob) is a vertex set $S$ of size $k$ such that $\cova(S) \ge \cova(S')$ (resp., $\cova(S) \le \cova(S')$) for every vertex set of size~$k$.
A graph $H$ is a \emph{minor} of $G$ if a graph isomorphic to $H$ can be obtained from $G$ by vertex and edge removals and edge contractions. Given a graph $H$, a family of graph $\mathcal{H}$ is said to be \emph{$H$-minor free} if there is no $G\in \mathcal{H}$ having $H$ as a minor. A graph $H$ is an \emph{apex graph} if $H$ can be made planar by the removal of a single vertex.  

We refer to the textbook of Cygan et al.~\cite{CyganFKLMPPS15} for an introduction to Parameterized Complexity and we refer to the paper of Marx~\cite{Marx08} for an introduction to the area of parameterized approximation. 

\section{FPT Approximation Algorithms} \label{sec:approx}

In this section, we design an FPT Approximation Schemes  for \maxprob as well as \minprob{} parameterized by $k$ and $\alpha$, assuming $\alpha > 0$. 

\begin{theorem} \label{thm:approx-max}
	For any $0 < \alpha \le 1$ and $0 < \varepsilon \le 1$, \maxprob and \minprob each admits an FPT-AS parameterized by $k$, $\varepsilon$ and $\alpha$.
	More specifically, given a graph $G = (V, E)$ and an integer $k$, there exists an algorithm that runs in time $\lr{\frac{k}{\varepsilon \alpha}}^{\Oh(k)} \cdot n^{\Oh(1)}$, and finds a set $S \subseteq V$ such that $\cova(S) \ge (1-\varepsilon) \cdot \cova(O)$ for \maxprob and $\cova(S) \le (1 + \varepsilon) \cdot \cova(O)$ for \minprob, where $O \subseteq V$ is an optimal solution.
\end{theorem}

For the case that $\OPT \coloneq \cova(O)$ is large, the following greedy argument will be helpful.

\begin{lemma}
	\label{lemma:approx-base}
	For \maxprob, let $S$ be the set of $k$ vertices with the largest degrees.
	Then, $\cova(S) \ge \OPT - 2k^2$.
	For \minprob, let $S$ be the set of $k$ vertices with the smallest degrees.
	Then, $\cova(S) \le \OPT + 2k^2$.
\end{lemma}
\begin{proof}
	Without loss of generality, we assume that $O \neq S$.
	Let $S \setminus O = \LR{y_1, y_2, \ldots, y_t}$, and $O \setminus S = \LR{w_1, w_2, \ldots, w_t}$, where $1 \le t \le k$.
	Here, we index the vertices so that $d(y_i) \ge d(y_j)$ and $d(w_i) \ge d(w_j)$ (for \minprob, $d(y_i) \le d(y_j)$ and $d(w_i) \le d(w_j)$) for $i < j$.
	Note that due to the choice of $S$, it holds that $d(y_i) \ge d(w_i)$ ($d(y_i) \le d(w_i)$ for \minprob) for each $1 \le i \le t$.
	
	Now we define a sequence of solutions $O_0, O_1, \ldots, O_t$, where $O_0 = O$, and for each $1 \le i \le t$,  $O_i \coloneqq (O_{i-1} \setminus \LR{w_i}) \cup \LR{y_i}$. Note that $O_t = S$.
	We claim that for each $1 \le i \le t$, $\cova(O_i) \ge \cova(O_{i-1}) -2k$ for \maxprob and $\cova(O_i) \le \cova(O_{i-1}) + 2k$ for \minprob. To this end, we note that $O_i$ is obtained from $O_{i-1}$ by removing $w_i$ and adding $y_i$. Thus, $\cova(O_i) = \cova(O_{i-1}) - (\alpha m_1 + ((1-\alpha) - \alpha) \cdot m_2) + \alpha m_3 + ((1-\alpha) - \alpha) \cdot m_4$, where
	\begin{align*}
		m_1 &\coloneqq m(\LR{w_i}, V \setminus O_{i-1}), & m_2 &\coloneqq m(\LR{w_i}, O_{i-1} \setminus \LR{w_i}), \\
		m_3 &\coloneqq m(\LR{y_i}, V \setminus O_{i}), & m_4 &\coloneqq m(\LR{y_i}, O_{i}\setminus \LR{w_i}).
	\end{align*}
	Observe that $d(w_i) - k \le m_1 \le d(w_i)$, $d(y_i) - k \le m_3 \le d(y_i)$, and $0 \le m_2, m_4 \le k$.
	We consider \maxprob first. We have that
	\begin{align*}
		\cova(O_i)
		&=  \cova(O_{i-1}) + \alpha (m_3 - m_1) + (1 - 2\alpha) (m_4 - m_2) \\
		&\ge \cova(O_{i-1}) + \alpha (m_3 - m_1) - |(1 - 2\alpha) (m_4 - m_2)|.
	\end{align*}
	Since $m_3 - m_1 \ge d(y_i) - d(w_i) - k \ge -k$ and $|(1 - 2\alpha) (m_4 - m_2)| \le k$,
	we obtain $\cova(O_i) \ge \cova(O_{i - 1}) - 2k$, regardless of the value of $\alpha$.
	We consider \minprob next. It holds that 
	\begin{align*}
		\cova(O_{i}) 
		&= \cova(O_{i-1}) + \alpha (m_3 - m_1) + (1 - 2\alpha) (m_4 - m_2) \\
		&\le \cova(O_{i-1}) + \alpha (m_3 - m_1) + |(1 - 2\alpha) (m_4 - m_2)|.
	\end{align*}
	Since $m_3 - m_1 \le d(y_i) - d(w_i) + k \le k$ and $|(1 - 2\alpha) (m_4 - m_2)| \le k$,
	we obtain $\cova(O_i) \le \cova(O_{i - 1}) + 2k$, regardless of the value of $\alpha$.
	
	Therefore, $\cova(O_t) \ge \cova(O_{0}) - 2kt \ge \OPT - 2k^2$ for \maxprob and $\cova(O_t) \le \cova(O_0) + 2kt \le \OPT + 2k^2$ for \minprob.
\end{proof}

\Cref{lemma:approx-base} allows us to find an approximate solution when $\OPT$ is sufficiently large.
The case that $\OPT$ is small remains.
We use different approaches for \maxprob and \minprob.

\paragraph*{Algorithm for \maxprob.}

Let $v_1$ be a vertex with the largest degree.
Our algorithm considers two cases depending on whether $d(v_1) > \Delta \coloneq \frac{2k^2}{\varepsilon \alpha} + k$.
If $d(v_1) > \Delta$, we can argue that the set $S$ from \Cref{lemma:approx-base} a $(1-\varepsilon)$-approximate solution.
To that end, we make the following observation.
\begin{observation} \label{obs:highdeg}
	If $d(v_1) > \Delta$, then $2k^2 \le \varepsilon \cdot \cova(S)$.
\end{observation}
\begin{proof}
	Note that $m(S, V\setminus S) = \sum_{u \in S} m(\LR{u}, V \setminus S) \ge m(\{v_1\}, V\setminus S) \ge d(v_1) - k$, where the inequality follows from the fact that at most $k$ edges incident to $v_1$ can have the other endpoint in $S$. This implies that 
	\begin{align*}
		\cova(S) \ge \alpha \cdot m(S, V \setminus S) \ge \alpha \cdot (d(v_1) - k) &\ge  \frac{2k^2}{\varepsilon}.
	\end{align*}
	Where we use the assumptions that $0 < \alpha \le 1$ and $d(v_1) \ge \Delta$.
\end{proof}
Thus, for $d(v_1) > \Delta$, we have $\OPT \le \cova(S) + 2k^2 \le (1 + \varepsilon) \cdot \cova(S)$, and thus $\cova(S) \ge (1 - \varepsilon) \cdot \OPT$.

So assume that $d(v_1) < \Delta$.
In this case, the maximum degree of the graph is bounded by $\Delta = \frac{2k^2}{\varepsilon \alpha} + k = \Oh(\frac{k^2}{\varepsilon \alpha})$. In this case, we solve the problem optimally using the algorithm of Shachnai and Zehavi \cite{SZ17} for \minprob, that runs in time $4^{k+o(k)} \cdot \Delta^k \cdot n^{\Oh(1)}$, which is at most $\lr{\frac{k^2}{\varepsilon\alpha}}^{\Oh(k)} \cdot n^{\Oh(1)}$. Combining both cases, we conclude the proof of \Cref{thm:approx-max} for \maxprob.

\paragraph*{Algorithm for \minprob.}

For \minprob, our algorithm considers two cases depending on the value of $\OPT$.
If $\OPT \ge \frac{2k^2}{\varepsilon}$, then our algorithm returns the set $S$ from \Cref{lemma:approx-base}.
Note that $\cova(S) \le \OPT + 2k^2 \le (1 + \varepsilon) \cdot \OPT$.

Now suppose that $\OPT < \frac{2k^2}{\varepsilon}$. In this case, we know that $O$ cannot contain a vertex of degree larger than $\Delta \coloneqq \frac{2k^2}{\alpha \varepsilon} + k$, for otherwise, $\cova(O) > \alpha (\Delta - k) \ge \OPT$, which is a contradiction. Thus, in this case the maximum degree of the graph is bounded by $\Delta$, and again we can solve the problem optimally in time $\lr{\frac{k^2}{\varepsilon\alpha}}^{\Oh(k)} \cdot n^{\Oh(1)}$, using the algorithm of Shachnai and Zehavi \cite{SZ17} for \minprob.

Since the value of $\OPT$ is unknown to us, we cannot directly conclude which case is applicable.
So we find a solution for each case and return a better one. This completes the proof of \Cref{thm:approx-max} for \minprob. 

\subsection{Faster FPT-AS for \maxprob when $\alpha \ge 1/3$}
In this section, we show that that a simpler idea of Manurangsi \cite{Manurangsi19} gives a faster FPT-AS for \maxprob when $\alpha \ge 1/3$, i.e., $\alpha \ge 1-2\alpha$, leading to the following theorem.

\begin{theorem} \label{thm:improved-fptas}
	For any $1/3 \le \alpha \le 1$, \maxprob admits an FPT-AS running in time $\lr{\frac{1}{\varepsilon}}^{\Oh(k)} \cdot n^{\Oh(1)}$.
\end{theorem}
\begin{proof}

Let $0 < \varepsilon < 1$ be fixed and let us sort the vertices of $V(G)$ by their degrees (breaking ties arbitrarily). Let $V' \subseteq V(G)$ denote the $k + \lceil\frac{4k}{\varepsilon^2} \rceil$ vertices of the largest degrees. We show that $V'$ contains a $(1-\varepsilon)$-approximate solution. Let $O$ denote an optimal solution for \maxprob. Further define $O_{i} \coloneqq O \cap V'$, $O_{o} \coloneqq O \setminus V'$.

Let $U \coloneqq V' \setminus O_{i}$ and let $U^* \subseteq U$ be a subset of size $|O_o|$ chosen uniformly at random from $U$. Let $\rho \coloneqq \frac{|O_o|}{|U|} \le \frac{k}{|U|} \le \varepsilon^2/4$.
In \Cref{lem:random}, we show that $\ex[\cova(O_{i} \cup U^*)] \ge (1-\varepsilon) \cdot \cova(O)$, which implies that $V'$ contains a $(1-\varepsilon)$-approximate solution. The algorithm simply enumerates all subsets of size $k$ from $V'$ and returns the best solution found. It follows that the running time of the algorithm is $\binom{|V'|}{k} \cdot n^{\Oh(1)} = \lr{\frac{1}{\varepsilon}}^{\Oh(k)} \cdot n^{\Oh(1)}$. All that remains is the proof of the following lemma.

\begin{lemma} \label{lem:random}
	$\ex[\cova(O_{i} \cup U^*)] \ge (1-\varepsilon) \cdot \cova(O)$.
\end{lemma}
\begin{proof}
	We fix some notation. For a vertex $u \in V$ and a subset $R \subseteq V$, we use $d_R(u)$ to denote the number of neighbors of $u$ in $R$. When $R = V$, we use $d(u)$ instead of $d_V(u)$. Let $S = O_{i} \cup U^* = (O \setminus O_{o}) \cup U^*$. We want to analyze the expected value of $\cova(S)$. To this end, we write $\cova(S) = \cova(O) - A + B$, where $A$ is the ``loss'' in the objective due to removal of $O_{out}$ and $B$ is the ``gain'' in the objective due to addition of $U^*$, defined as follows. 
	\begin{align*}
		A &= \alpha \cdot m(O_o, V\setminus O_o) + (1-\alpha) \cdot m(O_o)
		\\B &= Q_1 + Q_2 - \alpha \cdot m(O_i, U^*), \text{where, }
		\\Q_1 &= \alpha \cdot m(U^*, V \setminus (U^* \cup O_i)) + (1-\alpha) \cdot m(U^*) 
		\\Q_2 &= (1-\alpha) \cdot m(O_i, U^*)
	\end{align*}
	$Q_1$ is the total contribution of the edges with at least one endpoint in $U^*$ and other outside $S$, and $Q_2$ is the total contribution of edges with one endpoint in $U^*$ and other in $O_i$. Note that the lemma is equivalent to showing that $\ex[B - A] \ge - \varepsilon \cdot \cova(O)$, where the expectation is over the choice of $U^*$.
	
	Since $A$ does not depend on the choice of $U^*$, we have
	\begin{align}
		\ex[A] = A = \alpha \cdot m(O_o, V\setminus O_o) + (1-\alpha) \cdot m(O_o) 
		\le \alpha \cdot m(O_o, V\setminus O_o) + 2 \alpha \cdot m(O_o) = \alpha \sum_{v \in O_o} d(v) \label{eqn:1}
	\end{align}
	Here the first inequality follows from $\alpha \ge 1/3$.
	Now let us consider $\ex[B] = \ex[Q_1 + Q_2 - \alpha \cdot m(U^*, O_i)]$. For any pair of distinct vertices $u, v$, let $X_{uv} = 1$ if $\LR{u, v}$ is an edge and $X_{uv} = 0$ otherwise. Then, consider
	\begin{align}
		\ex[m(U^*, O_i)] &= \sum_{u \in U} \sum_{v \in O_i} X_{uv} \cdot \Pr(v \in U^*) = \rho \sum_{u \in U} \sum_{v \in O_i} X_{uv} \le \frac{\varepsilon^2}{4} \cdot m(O_i, U) 
		\label{eqn:2}
	\end{align}
	Now we analyze $\ex[Q_1]$. 
	For every edge with one endpoint in $U$ and the other in $V \setminus (U \cup O_i)$, there is a contribution $\alpha$ to $Q_1$ with probability $\rho$.
	Moreover, for every edge with both endpoints in $U$, the contribution to $Q_1$ is $\alpha$ with probability $2 \rho (1 - \rho)$ and $1 - \alpha$ with probability $\rho^2$.
	Thus, we obtain
	\begin{align}
		\ex[Q_1]
		&= \alpha \rho \cdot m(U, V \setminus (U \cup O_i)) + (2\alpha \rho(1 - \rho) + (1 - \alpha)\rho^2) \cdot m(U) \nonumber\\
		&\ge \alpha \rho \cdot m(U, V \setminus (U \cup O_i) + (2 \alpha \rho + (1 - 3 \alpha) \rho^2) \cdot m(U) \nonumber\\
		&= \alpha \rho \cdot (m(U, V \setminus (U \cup O_i)) + 2m(U)) = \alpha \rho \sum_{u \in U} d_{V \setminus O_i}(u).  \label{eqn:3}
	\end{align}
	Here the inequality is due to $\alpha \ge 1/3$.
	
	Note that for any $u \in U$, and $v \in O_o$, $d(u) \ge d(v)$, which implies that for any $u \in U$, $d(u) \ge \frac{\sum_{v \in O_o} d(v)}{|O_o|}$. Therefore, 	
	\begin{align}
		\sum_{u \in U} d(u) \ge \frac{|U|}{|O_o|} \sum_{v \in O_o} d(v) = \frac{1}{\rho} \cdot \sum_{v \in O_o}d(v) \label{eqn:31}
	\end{align}
	Now we consider two cases. 
	
	\textbf{Case 1:} $\sum_{u \in U} d(u) \le \frac{4}{\varepsilon} \cdot m(O_i, U)$. Then,
	\begin{align}
		&&\frac{4}{\varepsilon} \cdot m(O_i, U) &\ge \sum_{u \in U} d(u) \ge \frac{1}{\rho} \cdot \sum_{v \in O_o} d(v) \tag{Using (\ref{eqn:31})}
		\\&\implies &\frac{4}{\varepsilon} \cdot m(O_i, U) &\ge \frac{8}{\varepsilon^2} \cdot \sum_{v \in O_o} d(v) \tag{Since $\rho \le \varepsilon^2/4$}
		\\&\implies &\varepsilon/2 \cdot \alpha \cdot m(O_i, U) &\ge \alpha \sum_{v \in O_o} d(v) \ge \ex[A]  \label{eqn:41}
	\end{align}
	Where we use (\ref{eqn:1}) in the last inequality. Then consider,
	\begin{align}
		\ex[B- A] &\ge -\alpha \cdot \ex[m(U^*, O_i)] - \ex[A] \nonumber
		\\&\ge \frac{\varepsilon^2}{8} \alpha \cdot m(O_i, U) - \frac{\varepsilon}{2} \cdot \alpha \cdot m(O_i, U) \ge -\varepsilon \alpha \cdot m(O_i, U) \tag{Using (\ref{eqn:2}) and (\ref{eqn:41})}
		\\&\ge - \varepsilon \cdot \cova(O)
	\end{align}
	This finishes the first case. 
	
	\textbf{Case 2:} $\sum_{u \in U} d(u) > \frac{4}{\varepsilon} \cdot m(O_i, U)$. This implies that,
	\begin{align}
		\frac{\varepsilon}{4} \cdot \sum_{u \in U} d(u) &>  \sum_{u \in U} d_{O_i}(u) \nonumber
		\\\implies \sum_{u \in U} d_{V \setminus O_i}(u) &\ge \lr{1 - \frac{\varepsilon}{4}} \cdot \sum_{u \in U} d(u) \label{eqn:42}
	\end{align}
	Then, plugging back in (\ref{eqn:3}), we obtain,
	\begin{align}
		\ex[Q_1] &\ge \alpha \rho(1-\rho/2) \cdot (1-\varepsilon/4) \cdot \sum_{u \in U} d(u) \nonumber
		\\&\ge \alpha \rho (1-\varepsilon/2) \cdot \sum_{u \in U} d(u) \nonumber
		\\&\ge \alpha \rho (1-\varepsilon/2) \cdot \frac{|U|}{|O_o|} \sum_{v \in O_i} d(v) \tag{From \ref{eqn:31}} 
		\\&\ge \alpha (1-\varepsilon/2) \cdot \sum_{v \in O_i} d(v) \nonumber
		\\&\ge A \cdot (1-\varepsilon/2)  \label{eqn:4}
	\end{align}
	Where we use (\ref{eqn:1}) in the last inequality. Then, by (\ref{eqn:2}) and (\ref{eqn:4}), we obtain that,
	\begin{align}
		\ex[B - A] = \ex[B] - \ex[A] &\ge  -\varepsilon/2 \cdot \ex[A] - \alpha \varepsilon \cdot m(O_i, V\setminus O)  \label{eqn:5}
	\end{align}
	Now we argue that $\alpha \cdot m(O_i, V\setminus O) + \ex[A] = \alpha \cdot m(O_i, V \setminus O) + A \le \cova(O)$. All edges with one endpoint in $O_i$ and other outside $O$ contribute $\alpha$ to the objective, which corresponds to the first term. Note that $A$ is exactly the contribution of edges with at least one endpoint in $O_o$ to the objective. Further, note that no such edge has one endpoint in $O_i$ and other outside $O$, and thus is not counted in the first term. Thus, the sum of two terms is upper bounded by the objective, $\cova(O)$. Plugging it back in (\ref{eqn:5}), we obtain that $\ex[B - A] \ge -\varepsilon \cdot \cova(O)$ in the second case as well. This completes the proof of the lemma and the theorem.
\end{proof}
\end{proof}
\paragraph{Example showing a gap for $\alpha < 1/3$.} We now describe examples showing that for each fixed $\alpha < 1/3$, the above strategy of focusing on a bounded number of vertices of the largest degree does not lead to a $(1-\varepsilon)$-approximation, for large enough $k$. Let $f(k, \varepsilon)$ be an arbitrary function. Consider any $\alpha = 1/3 - \mu$, where $0 < \mu \le 1/3$, and let $N \ge f(k, \epsilon)$ be a large positive integer. The graph $G = (V, E)$ showing a gap is defined as follows. $V = H \uplus L \uplus O$, where $|H| = N, |L| = kN$ and $|O| = k$, thus $|V| = N(k+1) + k$. For each vertex $v \in H$, we attach $k$ distinct vertices from $L$ as pendants. Finally, we add all $\binom{k}{2}$ edges among the vertices of $O$, making in into a complete graph. 

Each vertex of $H$ has degree exactly $k$, each vertex of $O$ has degree exactly $k-1$, and each vertex of $L$ has degree exactly $1$, this gives the sorted order. It follows that the first $f(k, \varepsilon)$ vertices in the sorted order, say $T$, all belong to $N$. Furthermore, for any subset $S \subset T$ of size $k$, $\cova(S) = \alpha \cdot k^2 = (\frac{1}{3} - \mu) \cdot k^2$. On the other hand,  $\cova(O) = (1-\alpha) \cdot \binom{k}{2} = (\frac{2}{3} + \mu ) \cdot \frac{k^2 - k}{2} \approx (\frac{1}{3} + \mu) \cdot k^2$, assuming $k$ is large enough. Thus, any $k$-sized subset $S \subseteq T$, $\frac{\cova(S)}{\cova(O)} <  \frac{1/3 - \mu}{1/3 + \mu} \le 1 - 3\mu$. Thus, $T$ does not contain a $(1-\varepsilon)$-approximate solution for any $\varepsilon < 3\mu$. This shows that our analysis of \Cref{thm:improved-fptas} is tight for the range of $\alpha \ge 1/3$. 

\section{Subexponential FPT Algorithm for \maxprob on Apex-Minor Free Graphs} \label{sec:subexp}

Fomin et al.~\cite{FominLRS11} showed that \textsc{Partial Vertex Cover} on apex-minor free graphs can be solved in time $2^{\Oh(\sqrt{k})} \cdot n^{\Oh(1)}$.
In this section, we will prove its generalization to \maxprob as well as \minprob:

\begin{theorem} \label{thm:subexp}
	For an apex graph $H$, let $\cH$ be a family of $H$-minor free graphs.
	\begin{itemize}
		\item 
		For any $\alpha \ge 1/3$, \maxprob for $\cH$ can be solved in $2^{\Oh(\sqrt{k})} \cdot n^{\Oh(1)}$ time.
		\item
		For any $\alpha \le 1/3$, \minprob for $\cH$ can be solved in $2^{\Oh(\sqrt{k})} \cdot n^{\Oh(1)}$ time.
	\end{itemize}
\end{theorem}

We will give a proof for the maximization variant.
The minimization variant follows analogously.
Let $\sigma = v_1, v_2, \ldots, v_n$ be an ordering of vertices of $V$ in the non-increasing order of degrees, with ties broken arbitrarily. That is, $d(v_1) \ge d(v_2) \ge \ldots \ge d(v_{n-1}) \ge d(v_n)$. We will denote the graph by $G = (\vs, E)$ to emphasize the fact that the vertex set is ordered w.r.t. $\sigma$. We also let $\vs^j = \LR{v_1, \ldots, v_j}$. We first prove the following lemma.

\begin{lemma} \label{lem:exchange-lem-max}	
	Let $G = (\vs, E)$ be a yes-instance for \maxprob, where $1/3 \le \alpha \le 1$. Let $C = \LR{u_{i_1}, u_{i_2}, \ldots, u_{i_k}}$ be the lexicographically smallest solution for \maxprob and $u_{i_k} = v_j$ for some $j$. Then $C$ is a dominating set of size $k$ for $G[\vs^j]$.
\end{lemma}
\begin{proof}
	Suppose for the contradiction that $C$ is not a dominating set for $G[\vs^j]$. Then, there exists a vertex $v_i$ with $1 \le i < j$ such that $N[v_i] \cap C = \emptyset$. Set $C' = (C \setminus \LR{v_j}) \cup \LR{v_i}$. Note that $d(v_i) \ge d(v_j)$. Define the following:
	\begin{align*}
		m_1 &= m(\LR{v_j}, V \setminus C),
		\\m_2 &= m(\LR{v_j}, C \setminus \LR{v_j}),
		\\m_3 &= m(\LR{v_i}, (V \setminus C)\cup \LR{v_j}) = d(v_i),
		\\m_4 &= m(\LR{v_i}, C \setminus \LR{v_j}) = 0.
	\end{align*}
	We will show that $C'$ is another solution for the \maxprob instance.
	Since $C' \setminus \LR{v_i} = C \setminus \LR{v_j}$, it suffices to show that
	\begin{align*}
		\cova(C') - \cova(C) = (\cova(C') - \cova(C' \setminus \LR{v_i})) - (\cova(C) - \cova(C \setminus \LR{v_j}))
	\end{align*}
	is nonnegative.
	By definition,
	\begin{align}
		&\cova(C') - \cova(C' \setminus \LR{v_i}) = \alpha \cdot m_3 + ((1 - \alpha) - \alpha) \cdot m_4 = \alpha \cdot d(v_i) \text{ and } \nonumber\\ 
		&\cova(C) - \cova(C \setminus \LR{v_j}) = \alpha \cdot m_1 + ((1-\alpha) - \alpha) \cdot m_2 \le \alpha \cdot (m_1 + m_2) = \alpha \cdot d(v_j), \label{eqn:bound}
	\end{align}
	where the inequality is due to the assumption that $\alpha \ge 1/3$.	Therefore, 
	\begin{align*}
		\cova(C') - \cova(C)
		= \alpha \cdot (d(v_i) - d(v_j)) \ge 0,
	\end{align*}
	which is a contradiction to the assumption that $C$ is the lexicographically smallest solution for \maxprob.
\end{proof}

In view of \Cref{lem:exchange-lem-max}, we can use the following approach to search for the lexicographically smallest solution $C$.
First, we guess the last vertex $v_j$ of $C$ in the ordering $\sigma$, i.e.,
we search for a solution $C$ such that $v_j \in C$ and $C \subseteq V_{\sigma}^{j}$.
If $G[V_{\sigma}^j]$ has no dominating set of size at most, say $2k$, then we reject.
This can be done in polynomial time, since \textsc{Dominating Set} admits a PTAS on apex-minor free graphs \cite{DemaineH05}.
We thus may assume that there is a dominating set of size $2k$ in $G[V_{\sigma}^j]$.
It is known that an apex-minor free graph with a dominating set of size $\kappa$ has treewidth $\Oh(\sqrt{\kappa})$, where $\Oh$ hides a factor depending on the apex graph whose minors are excluded \cite{DemaineFHT04,DemaineH08,FominGT09}.
We can use a constant-factor approximation algorithm of Demaine \cite{DemaineHK05} to find a tree decomposition $\mathcal{T}$ of width $w \in \Oh(\sqrt{k})$.
Finally, we solve the problem via dynamic programming over the tree decomposition.
Bonnet et al.~\cite{BEPT15} gave a  $\Oh^*(2^w)$-time algorithm that solves \maxprob with a tree decomposition of width $w$ given.
We need to solve a slightly more general problem because $\mathcal{T}$ is the tree decomposition is over $V_{\sigma}^j$.
To remove $V \setminus V_{\sigma}^j$, we introduce a weight $\omega \colon V_{\sigma}^j \to \mathbb{N}$ defined by $\omega(v) = |N(v) \cap (V \setminus V_{\sigma}^j)|$.
The objective is then to maximize $\cova(C) + \alpha \sum_{v \in C} \omega(C)$.
The dynamic programming algorithm of Bonnet et al. can be adapted to solve this weighted variant in the same running time.
Thus, we obtain a $2^{\Oh(\sqrt{k})} \cdot n^{\Oh(1)}$-time algorithm for \maxprob.

For \minprob, we can show the following lemma whose proof is omitted because it is almost analogous to the previous one. The only change is that, $\vs$ refers to the vertices in the non-decreasing order of degrees. Also, we consider the regime where $0 \le \alpha \le 1/3$, which implies $\alpha \le 1-2\alpha$, which would give the reverse inequality in (\ref{eqn:bound}). 
\begin{lemma} \label{lem:exchange-lem-min}
	Let $G = (\vs, E)$ be a yes-instance for \maxprob, where $0 \le \alpha \le 1/3$. Let $C = \LR{u_{i_1}, u_{i_2}, \ldots, u_{i_k}}$ be the lexicographically smallest solution for \maxprob and $u_{i_k} = v_j$ for some $j$. Then $C$ is a dominating set of size $k$ for $G[\vs^j]$.
\end{lemma}

With this lemma at hand, an analogous algorithm solves \minprob in $2^{\Oh(\sqrt{k})} \cdot n^{\Oh(1)}$ time, thereby proving \Cref{thm:subexp}.

\section{Conclusion} \label{sec:conclusion}
In this paper, we demonstrated that the algorithms exploiting the ``degree-sequence'' that have been successful for designing algorithms for \textsc{Max $k$-Vertex Cover} naturally generalize to $\textsc{Max/}\minprob$. Specifically, we designed FPT approximations for $\textsc{Max}/\minprob$ parameterized by $k, \alpha$, and $\varepsilon$, for any $\alpha \in (0, 1]$. For \maxprob, this result is tight since, when $\alpha = 0$, the problem is equivalent to \textsc{Densest $k$-Subgraph}, which is hard to approximate in FPT time \cite{ManurangsiRS21}. We also designed subexponential FPT algorithms for \maxprob (resp.\ \minprob) for the range $\alpha \ge 1/3$ (resp.\ $\alpha \le 1/3$) on any apex-minor closed family of graphs. It is a natural open question whether one can obtain subexponential FPT algorithms for $\textsc{Max}/\minprob$ for the entire range $\alpha \in [0, 1]$. A notable special case is that of \textsc{Densest $k$-Subgraph} on planar graphs. In this case, the problem is not even known to be \textsf{NP}-hard, if the subgraph is allowed to be disconnected. For the \textsc{Densest Connected $k$-Subgraph} problem, it was shown by Keil and Brecht~\cite{KeilB91} that the problem is NP-complete on planar graphs. From the other side, it can be shown that \textsc{Densest Connected $k$-Subgraph}  admits a subexponential in $k$ randomized algorithm on apex-minor free graphs using the general results of Fomin et al.~\cite{FominLMPPS22}. Thus, dealing with disconnected dense subgraphs is difficult for both algorithms and lower bounds.

\bibliography{references}

\end{document}